\documentclass[conference]{IEEEtran}
\IEEEoverridecommandlockouts

\usepackage{amsmath}
\usepackage{amssymb}
\usepackage{amsthm}
\usepackage{dsfont}
\usepackage{enumitem}
\usepackage{hyperref}

\usepackage{graphicx}
\newcommand{\E}[3]{\mathbb{E}^{#2} _{#1}  \left[ #3 \right]}

\newcommand{\Prob}[3]{\mathbb{P}^{#2} _{#1} \left\{ #3 \right\}}
\newcommand{\qqquad}{\qquad \quad}

\newcommand{\mc}{\mathcal}

\newcommand{\td}{\Tilde}
\newcommand{\KL}[2]{D(#1||#2)}

\newcommand{\FAR}[2]{\text{FAR} ^{#1} \left( #2 \right)}
\newcommand{\WADD}[2]{\text{WADD} ^{#1} \left( #2 \right)}

\DeclareMathOperator*{\esssup}{ess\,sup}

\newtheorem{theorem}{Theorem}[section]

\newtheorem{lemma}[theorem]{Lemma}

\theoremstyle{definition}

\theoremstyle{remark}
\newtheorem*{remark}{Remark}    

\begin{document}

\title{Non-Parametric Quickest Detection of a Change in the Mean of an Observation Sequence\\
}

\author{\IEEEauthorblockN{Yuchen Liang}
\IEEEauthorblockA{\textit{ECE Department and Coordinated Science Lab} \\
\textit{University of Illinois at Urbana-Champaign}\\
Urbana, IL \\
yliang35@illinois.edu}
\and
\IEEEauthorblockN{Venugopal V. Veeravalli}
\IEEEauthorblockA{\textit{ECE Department and Coordinated Science Lab} \\
\textit{University of Illinois at Urbana-Champaign}\\
Urbana, IL \\
vvv@illinois.edu}
\thanks{This work  was supported in part by the National Science Foundation under grant  ECCS-2033900, and by the Army Research Laboratory under Cooperative Agreement W911NF-17-2-0196, through the University of Illinois at Urbana-Champaign.}
}

\maketitle

\begin{abstract}
We study the problem of quickest detection of a change in the mean of an observation sequence, under the assumption that both the pre- and post-change distributions have bounded support. We first study the case where the pre-change distribution is known, and then study the extension where only the mean and variance of the pre-change distribution are known. In both cases, no knowledge of the post-change distribution is assumed other than that it has bounded support. For the case where the pre-change distribution is known, we derive a test that asymptotically minimizes the worst-case detection delay over all post-change distributions, as the false alarm rate goes to zero.  We then study the limiting form of the optimal test  as the gap between the pre- and post-change means goes to zero, which we call the Mean-Change Test (MCT). We show that the MCT can be designed with only knowledge of the mean and variance of the pre-change distribution. We validate our analysis through numerical results for detecting a change in the mean of a beta distribution. We also demonstrate the use of the MCT for pandemic monitoring.

\end{abstract}


\section{Introduction}
Quickest Change Detection (QCD) is a fundamental problem in mathematical statistics. Given a stochastic sequence whose distribution changes at some unknown change-point, the goal is to detect the change after it occurs as quickly as possible, subject to false alarm constraints. The QCD framework has had a wide range of applications, including but not limited to line-outage in power systems \cite{line-outage}, dim-target manoeuvre detection \cite{Molloy2017}, stochastic process control \cite{spc}, structural health monitoring \cite{shm}, and, more recently, in the Multi-Armed Bandit (MAB) problem with piece-wise stationary bandits \cite{kaufmann}. Two formulations are used in the classical QCD problem: the Bayesian formulation \cite{bayesqcd}, where the change-point is assumed to follow some prior distribution, and the minimax formulation \cite{lorden1971,pollak1987}, where the worst-case detection delay is minimized over all possible change-points, subject to false alarm constraints. In both the Bayesian and minimax settings, if the pre- and post-change distributions are known, low-complexity efficient solutions to the QCD problem can be found \cite{moulin-veeravalli-2018}.

In many practical situations, we may not know the exact distribution in the pre- or post-change regimes. Although it is reasonable to assume that we can obtain a large amount of data in the pre-change regime, this may not be the case for the post-change regime. In applications such as the epidemic detection problem and the MAB problem, a change in a specific statistic (e.g., the mean) of the distribution is of interest. This is different from the original QCD problem where any distributional change needs to be detected. Furthermore, in many applications, the support of the distribution is bounded. For example, the observations representing the fraction of some specific group in the entire population are bounded between 0 and 1. This is the case, for example, in the pandemic monitoring problem that we discuss in detail in Section IV.

There have been a number of lines of work on the QCD problem when the pre- and/or post-change distributions are not completely known. The most prevalent is the Generalized Likelihood Ratio (GLR) approach, in which the maximum of the test statistics corresponding to each possible change-points is used to construct the test. A GLR approach for a Gaussian mean-change detection problem is proposed in \cite{siegmund1995}. A GLR test for the case where the pre- and post-change distributions come from one-parameter exponential families is analyzed in \cite{LaiXing2010}. A kernel approach to construct GLR statistics is discussed in \cite{mmd}. More recently, a GLR mean-change test for sub-Gaussian observations using scan statistics is proposed and analyzed in \cite{maillard19a}, and this approach is applied to the MAB problem in \cite{kaufmann}. In all of these methods, the complexity of computing the test statistic at each time-step grows at least linearly with the number of samples. In practice, a windowed version of the GLR test statistic is used to keep the complexity small, while suffering some loss in performance. An alternative to the GLR approach was given in \cite{binning}, where a histogram is used to estimate the post-change distribution, and a low complexity test with a recursive structure is constructed.

Another line of work is the one based on a \textit{minimax} robust approach \cite{huber1965}, in which it is assumed that the distributions come from mutually exclusive uncertainty classes. Under certain conditions on the uncertainty classes, e.g., joint stochastic boundedness \cite{moulin-veeravalli-2018}, a low-complexity saddle-point solution to the minimax robust QCD problem can be found. In this paper, we use an asymptotic version of the minimax robust QCD problem formulation to develop algorithms for the non-parametric detection of a change in mean of an observation sequence. Our contributions are as follows:
\begin{enumerate}
    \item We study the problem of quickest detection of a change in the mean of an observation sequence under the assumption that no knowledge of the post-change distribution is available other than that it has bounded support.

    \item For the case where the pre-change distribution is known, we derive a test that asymptotically minimizes the worst-case detection delay over all possible post-change distributions, as the false alarm rate goes to zero.
    
\item     We study the limiting form of the optimal test as the gap between the pre- and post-change means goes to zero, which we call the Mean-Change Test (MCT). We show that the MCT can be designed with only knowledge of the mean and variance of the pre-change distribution. 

\item We validate our analysis through numerical results for detecting a change in the mean of a beta distribution. We also demonstrate the use of the MCT for pandemic monitoring.
  \end{enumerate}

\section{Problem Statement}
\label{sec:prob-stat}

Let $X_1,\dots,X_t,\dots \in [0,1]$ be independent samples, with $X_1,\dots,X_\nu \sim P_0$, and $X_{\nu+1},\dots \sim P_1$. Let $\Prob{\nu}{P_0,P_1}{\cdot}$ denote the probability measure on the entire sequence of observations when the pre- and post-change distributions are $P_0$ and $P_1$, respectively, and the change-point is at $t=\nu$, and  let $\E{\nu}{P_0,P_1}{\cdot}$ denote the corresponding expectation. We will first consider the case where $P_0$ is completely known, and then study extensions to the cases where only partial (moment) information about $P_0$ is available. The change-time $\nu$ is unknown but deterministic. The problem is to detect the change quickly while not causing too many false alarms. Let $\tau$ be a stopping time \cite{moulin-veeravalli-2018} defined on the observation sequence associated with the detection rule, i.e. $\tau$ is the time at which we stop taking observations and declare that the change has occurred. If both the pre- and post-change distributions are known, Lorden \cite{lorden1971} proposed solving the following optimization problem to find the best stopping time $\tau$:
\begin{equation}
    \inf_{\tau \in \mathcal{C}_\alpha} \WADD{P_0,P_1}{\tau}
\end{equation}
where
\begin{multline}
\label{orig_qcd}
    \WADD{P_0,P_1}{\tau} :=\\
    \sup_{\nu \geq 1} \esssup \E{\nu}{P_0,P_1}{\left(\tau-\nu+1\right)^+|X_1,\ldots, X_{\nu-1}}
\end{multline}
is a worst-case delay metric and
\begin{equation}
    \mathcal{C}^{\{P_0\}}_\alpha := \left\{ \tau: \FAR{P_0}{\tau} \leq \alpha \right\}
\end{equation}
is the feasible set where $\FAR{P_0}{\tau} := \E{\infty}{P_0,P_1}{\tau}^{-1}$. Note that $\E{\infty}{P_0,P_1}{\cdot}$ is the expectation operator when the change never happens, and $(\cdot)^+:=\max\{0,\cdot\}$.

Lorden also showed that Page's Cumulative Sum (CUSUM) algorithm \cite{page1954} whose test statistic is given by:
\begin{equation}
\label{cusum_stat}
\begin{split}
    \Lambda^{P_0,P_1}(t) &= \max_{1\leq k \leq t} \sum_{i=k}^t \ln L^{P_0,P_1}(X_i)\\
    &= \left(\Lambda^{P_0,P_1}(t-1) + \ln L^{P_0,P_1}(X_t) \right)^+
\end{split}
\end{equation}
with $\Lambda^{P_0,P_1}(0) = 0$ solves the problem in (\ref{orig_qcd}) asymptotically, where $L^{P_0,P_1}(x) = p_1(x)/p_0(x)$ is the likelihood ratio between the densities. The CUSUM stopping rule is given by:
\begin{equation}
\label{defstoppingrule}
    \tau\left(\Lambda^{P_0,P_1}, b_\alpha\right) := \inf \{t:\Lambda^{P_0,P_1}(t)\geq b_\alpha \}
\end{equation}
where $b_\alpha := |\ln \alpha|$. It was shown by Moustakides \cite{moustakides1986} that Page's test is exactly optimal for the problem in (\ref{orig_qcd}).

When the pre-change and post-change distributions are unknown but belong to some uncertainty sets, a minimax robust metric can be applied:
\begin{equation}
\label{mainprob}
    \inf_{\tau \in \mathcal{C}_{\alpha}^{\mc{P}_0}} \sup_{(P_0, P_1)\in \mathcal{P}_0 \times \mathcal{P}_1} \WADD{P_0,P_1}{\tau}
\end{equation}
where the feasible set becomes
\begin{equation}
\label{main_fa_constraint}
    \mathcal{C}_{\alpha}^{\mc{P}_0} = \left\{ \tau: \sup_{P_0 \in \mathcal{P}_0} \FAR{P_0}{\tau} \leq \alpha \right\}
\end{equation}

A pair of uncertainty sets $(\mc{P}_0,\mc{P}_1)$ is said to be \textit{jointly stochastically (JS) bounded} by $(\Bar{P}_0,\Bar{P}_1) \in \mc{P}_0 \times \mc{P}_1$ if, for any $(P_0,P_1) \in \mc{P}_0 \times \mc{P}_1$ and any $t > 0$,
\begin{align}
    P_0 \{L^{\Bar{P}_0,\Bar{P}_1}(X) > t\} &\leq \Bar{P}_0 \{L^{\Bar{P}_0,\Bar{P}_1}(X) > t\}\\
    P_1 \{L^{\Bar{P}_0,\Bar{P}_1}(X) > t\} &\geq \Bar{P}_1 \{L^{\Bar{P}_0,\Bar{P}_1}(X) > t\}
\end{align}
where $L^{\Bar{P}_0,\Bar{P}_1}$ is the likelihood ratio between $\Bar{P}_1$ and $\Bar{P}_0$ as defined previously \cite{moulin-veeravalli-2018}. The distributions $\Bar{P}_0$ and  $\Bar{P}_1)$ are called least favorable distributions (LFDs) within the classes ${\cal P}_0$ and ${\cal P}_1$, respectively. If the pair of pre- and post-change uncertainty sets is JS bounded, the test statistic $\Lambda^{\Bar{P}_0,\Bar{P}_1}(t)$ with the stopping rule $\tau(\Lambda^{\Bar{P}_0,\Bar{P}_1},b_\alpha)$ solves (\ref{mainprob}) exactly \cite{Unnikrishnan2011}.

A pair of uncertainty sets $(\mc{P}_0,\mc{P}_1)$ is said to be \textit{weakly stochastically (WS) bounded} by $(\td{P}_0,\td{P}_1) \in \mc{P}_0 \times \mc{P}_1$ if
\begin{equation}
\label{wsb:kl}
    \KL{\td{P}_1}{\td{P}_0} \leq \KL{P_1}{\td{P}_0} - \KL{P_1}{\td{P}_1}
\end{equation}
for all $P_1 \in \mc{P}_1$, and
\begin{equation}
\label{wsb:exp}
    \E{}{P_0}{L^{\td{P}_0,\td{P}_1}(X)} \leq \E{}{\td{P}_0}{L^{\td{P}_0,\td{P}_1}(X)} = 1
\end{equation}
for all $P_0 \in \mc{P}_0$ \cite{Molloy2017}, where $\E{}{P}{\cdot}$ denotes the expectation operator with respect to distribution $P$. It is shown in \cite{Molloy2017} that if a pair of uncertainty sets is JS bounded by $(\Bar{P}_0,\Bar{P}_1)$, it is also WS bounded by $(\Bar{P}_0,\Bar{P}_1)$.
% It is shown that if the pair of pre- and post-change uncertainty sets is WS bounded by $(\td{P}_0,\td{P}_1)$, the test statistic $\Lambda^{\td{P}_0,\td{P}_1}(t)$ with the stopping rule $\tau(\Lambda^{\td{P}_0,\td{P}_1},|\ln \alpha|)$ solves (\ref{mainprob}) asymptotically as $\alpha \to 0$ \cite{Molloy2017}. The LFDs in the sense of WS boundedness can be found using the following theorem:
\begin{theorem}
\label{wsb}
\cite{Molloy2017}\  If $(\td{P}_0,\td{P}_1)$ solves
\begin{equation}
\label{wsb:minkl}
    \inf_{(P_0,P_1) \in \mc{P}_0 \times \mc{P}_1} \KL{P_1}{P_0}
\end{equation}
and if furthermore, $\mc{P}_1$ is convex, and (\ref{wsb:exp}) holds for all $P_0 \in \mc{P}_0$, then $(\mc{P}_0, \mc{P}_1)$ is WS bounded by $(\td{P}_0,\td{P}_1)$, and $\tau(\Lambda^{\td{P}_0,\td{P}_1},b_\alpha)$ solves the problem in (\ref{mainprob}) asymptotically as $\alpha \to 0$. Also, the worst-case delay is
\begin{equation}
\small
    \inf_{\tau \in \mathcal{C}_{\alpha}^{\mc{P}_0}} \sup_{(P_0, P_1)\in \mathcal{P}_0 \times \mathcal{P}_1} \WADD{P_0,P_1}{\tau} = \frac{|\ln\alpha|}{\KL{\td{P}_1}{\td{P}_0}}(1+o(1))
\end{equation}
\end{theorem}

In this paper, we consider the problem of quickest change detection for observations with bounded support. The goal is to construct a test to detect when the mean of the observations exceeds some pre-specified threshold, i.e., $P_1 \in \mathcal{P}_1$ such that
\begin{equation}
\label{f1constraint}
    \mathcal{P}_1 := \{ P: \E{}{P}{X} \geq \eta > \mu_0\}
\end{equation}
In this expression, $X$ denotes a generic observation in the sequence, $\eta$ is a pre-designed threshold on the mean, and $\mu_0 := \E{}{P_0}{X}$. Define 
\begin{equation}
    \Delta := \frac{\eta-\mu_0}{2}
\end{equation}
Also, for $i=0,1$, let the density (w.r.t. the Lebesgue measure) of $P_i$ be $p_i$. Define
\begin{equation}
    \kappa_0 (\lambda) = \ln \E{}{P_0}{e^{\lambda X}}
\end{equation}
to be the cumulant-generating functiong (cgf) under measure $P_0$. We present our main results below.

\section{Main Results}
\label{sec:main-results}
\subsection{Known Pre-change Distribution}

Throughout we will assume that  $\mc{P}_1$ is as defined in (\ref{f1constraint}). In this subsection, we study the case where $\mc{P}_0 = \{ P_0 \}$.

\begin{theorem}
For ${\cal P}_0 = P_0$, and ${\cal P}_1$ given in (\ref{f1constraint}), define
\begin{equation}
\label{tilteddistr}
    p_1^*(x) = p_0(x) e^{\lambda^* x - \kappa_0 (\lambda^*)}
\end{equation}
where $\kappa_0 (\lambda)$ is the cgf under $P_0$ and $\lambda^*$ satisfies
\begin{equation}
\label{lambdastar}
    \kappa_0'(\lambda^*) := \frac{\E{}{P_0}{X e^{\lambda^* X}}}{\E{}{P_0}{e^{\lambda^* X}}} = \eta
\end{equation}
Then, the statistic
\begin{equation}
\label{opt_stat}
    \Lambda^{P_0,P_1^*}(t) = \max_{1\leq k \leq t} \sum_{i=k}^t \left(\lambda^* X_i - \kappa_0(\lambda^*) \right)
\end{equation}
and the stopping rule $\tau(\Lambda^{P_0,P_1^*},b_\alpha)$ with threshold $b_\alpha=|\ln{\alpha}|$ solves the minimax robust problem in (\ref{mainprob}) asymptotically as $\alpha \to 0$, and
\begin{equation}
\small
    \inf_{\tau \in \mathcal{C}_{\alpha}^{\mc{P}_0}} \sup_{(P_0, P_1)\in \mathcal{P}_0 \times \mathcal{P}_1} \WADD{P_0,P_1}{\tau} = \frac{|\ln{\alpha}|}{\lambda^* \eta - \kappa_0(\lambda^*)} (1+o(1))
\end{equation}
\end{theorem}

% \begin{lemma}
% \label{mainlemma}
% For ${\cal P}_0 = P_0$, and ${\cal P}_1$ given in (\ref{f1constraint}), the KL divergence minimization problem in (\ref{wsb:minkl}) has the solution:
% \begin{equation}
% \label{tilteddistr}
%     p_1^*(x) = p_0(x) e^{\lambda^* x - \kappa_0 (\lambda^*)}
% \end{equation}
% where $\kappa_0 (\lambda)$ is the cgf under $P_0$ and $\lambda^*$ satisfies
% \begin{equation}
% \label{lambdastar}
%     \kappa_0'(\lambda^*) := \frac{\E{}{P_0}{X e^{\lambda^* X}}}{\E{}{P_0}{e^{\lambda^* X}}} = \eta
% \end{equation}
% \end{lemma}

\begin{proof}
We follow the procedure outlined in \cite[Sec.~6.4.1]{levy-detection}. We want to minimize $\KL{P_1}{P_0} = \E{}{P_1}{\ln (p_1(x)/p_0(x))}$ subject to $\E{}{P_1}{X} \geq \eta$. We consider the Lagrangian
\begin{equation}
\begin{split}
    L(p_1,\lambda,\mu) &= \E{}{P_1}{\ln (p_1(x)/p_0(x))} \\
    &\qquad + \lambda (\eta - \E{}{P_1}{X}) + \mu \left(1 - \int_{[0,1]} p_1(x) d x \right)\\
    &= \int_{[0,1]} \left(\ln \frac{p_1(x)}{p_0(x)} - \lambda x - \mu \right) p_1(x) d x \\
   &\qqquad + \lambda \eta + \mu 
\end{split}
\end{equation}
where the Lagrange multiplier $\lambda \geq 0$ corresponds to the constraint that the post-change mean is greater than $\eta$, and $\mu$ corresponds to the constraint that $p_1(x)$ is a probability measure. 
%We do not explicitly impose the non-negativity constraint, because it turns out that this constraint is satisfied automatically. 
For an arbitrary direction $z$, since $p_1$ is continuous by assumption, we take the Gateaux derivative with respect to $p_1$:
\begin{equation}
    \begin{split}
        \nabla_{p_1,z} L(p_1,\lambda,\mu) &:= \lim_{h \to 0} \frac{L(p_1+h z,\lambda,\mu)-L(p_1,\lambda,\mu)}{h}\\
        &= \int_{[0,1]} \left(\ln \frac{p_1(x)}{p_0(x)} - \lambda x - \mu'\right) z d x
    \end{split}
\end{equation}
where $\mu' = \mu - 1$, and since $z$ is arbitrary, we arrive at
\begin{equation}
    \ln \frac{p_1(x)}{p_0(x)} - \lambda x - \mu' = 0
\end{equation}
By the Generalized Kuhn–Tucker Theorem \cite{luenburger-opt-by-vec-space}, if $p_0 (x)$ is bounded, $p_1 (x) = p_0(x) e^{\lambda x+\mu}$ is a necessary condition for optimality. Furthermore, since $L(p_1,\lambda,\mu)$ is convex in $p_1$, this is also a global optimum. To satisfy the constraints,
\begin{equation}
    \mu' = -\ln \int_{[0,1]} p_0 e^{\lambda x} d x = - \kappa_0 (\lambda)
\end{equation}
and $\lambda^*$ satisfies
\begin{equation}
    \eta = \E{}{P_1}{X} = \frac{\E{}{P_0}{X e^{\lambda^* X}}}{\E{}{P_0}{e^{\lambda^* X}}} = \kappa_0'(\lambda^*)
\end{equation}
Thus, $P_1^*$ in (\ref{tilteddistr}) minimizes the KL divergence when $P_0$ is known. By Theorem \ref{wsb}, since $(P_0,P_1^*)$ minimizes the KL divergence, $\mc{P}_1$ is convex in that the expectation operator is linear, and $\mc{P}_0$ is a singleton, the uncertainty classes $\mc{P}_0 \times \mc{P}_1$ are WS bounded by $(P_0,P_1^*)$. Therefore, an asymptotically optimal test is $\Lambda^{P_0,P_1^*}$ as in (\ref{opt_stat}).

Furthermore, the minimum KL-divergence is
\begin{equation}
\begin{split}
    \KL{P_1^*}{P_0} &= \int_{[0,1]} (\lambda^* x - \kappa_0(\lambda^*)) p_1^* (x) d x\\
    &= \lambda^* \eta - \kappa_0(\lambda^*)
\end{split}
\end{equation}
Hence, the worst-case delay satisfies
\begin{equation}
\small
\begin{split}
    \inf_{\tau \in \mathcal{C}^{\{P_0\}}_{\alpha}} \sup_{P_1\in \mathcal{P}_1} \WADD{P_0,P_1}{\tau} &= \frac{|\ln{\alpha}|}{\KL{P_1^*}{P_0}}(1+o(1))\\
    &= \frac{|\ln{\alpha}|}{\lambda^* \eta - \kappa_0(\lambda^*)}(1+o(1))
\end{split}
\end{equation}
as $\alpha \to 0$. \qedhere
\end{proof}

Since $p_0$ is a density on $[0,1]$, $p^*$ is also a density on $[0,1]$. Indeed, $p^*$ is the exponentially-tilted distribution (or the Esscher transform) of $p_0$. 
%This lemma implies that the least favorable post-change distribution is dependent on the pre-change distribution $P_0$.

\subsection{Approximation for Small $\Delta$}

Even though we have an expression for the test statistic when $P_0$ is known, as given in (\ref{opt_stat}), the exact solution of $\lambda^*$ is not available in closed-form. Fortunately, if the mean-change gap $\Delta$ is small, we obtain a low-complexity test in terms of only the pre-change mean and variance that closely approximates the performance of the asymptotically minimax optimal test in the previous section.
% $\tau(\td{\Lambda}^{\mu_0,\eta},\Tilde{b}_\alpha |\ln \alpha|\sigma_0^2)$ solves (\ref{mainprob}) asymptotically as $\alpha \to 0$ and $\Delta \to 0$. A similar form of the stopping rule as in (\ref{defstoppingrule}) is applied, but the thresholds $\td{b}(\alpha)$ depend on the constraints on the pre-change distribution. In the following analysis, we first study the case where $P_0$ is known. In this case, the pre-change uncertainty set only contains a singleton. Then, under the assumption that the gap $\Delta$ is asymptotically small, we convert the test statistic to (\ref{stat}) and prove that it is asymptotically optimal as $\alpha \to 0$. Next, we study the case where the pre-change class is $\mathcal{P}_0$. Finally, we study the case where $\Delta$ is moderate. In all cases, thresholds that match with the optimality results are given.

As the gap is asymptotically small, the worst-case post-change mean $\eta \to \mu_0$. Hence, $\lambda^* \to 0$. From the Taylor expansion on $\kappa_0$ around 0, we obtain
\begin{equation}
\label{taylor}
\begin{split}
    \kappa_0(\lambda^*) &= \kappa_0(0) + \kappa_0'(0) \lambda^* + \frac{\kappa_0''(0)}{2} (\lambda^*)^2 + o((\lambda^*)^2)\\
    &= \mu_0 \lambda^* + \frac{\sigma_0^2}{2} (\lambda^*)^2 + o((\lambda^*)^2)
\end{split}
\end{equation}
In this same regime, by continuity of $\kappa_0'(\cdot)$,
\begin{equation}
\begin{split}
    \lambda^* &= \frac{\kappa_0'(\lambda^*) - \kappa_0'(0)}{\kappa_0''(0)} + o(\Delta)\\
    &= \frac{\eta - \mu_0}{\sigma_0^2} + o(\Delta)\\
    &= \frac{2 \Delta}{\sigma_0^2} + o(\Delta)
\end{split}
\end{equation}
where we have used $\kappa_0'(\lambda^*) = \eta$. Hence, the approximate test statistic is
\begin{equation}
\begin{split}
    \lambda^* X_t - \kappa_0(\lambda^*) &= \lambda^* X_t - (\mu_0 \lambda^* + \frac{\sigma_0^2}{2} (\lambda^*)^2) + o((\lambda^*)^2)\\
    &= \frac{2 \Delta}{\sigma_0^2} (X_t - \mu_0) - \frac{\sigma_0^2}{2} \left(\frac{2 \Delta}{\sigma_0^2}\right)^2 + o(\Delta^2)\\
    &= \frac{2 \Delta}{\sigma_0^2} \left(X_t - \frac{\mu_0 + \eta}{2}\right) + o(\Delta^2)
\end{split}
\end{equation}
and the corresponding minimum KL-divergence is approximated as:
\begin{equation}
    \KL{P_1^*}{P_0} = \frac{2\Delta^2}{\sigma_0^2} + o(\Delta^2)
\end{equation}

Let $\sigma_0^2$ be the variance of $X$ under $P_0$. Since
\begin{equation}
\small
    \frac{2 \Delta}{\sigma_0^2} \left(X_t - \frac{\mu_0 + \eta}{2}\right) > b_\alpha \iff X_t - \frac{\mu_0 + \eta}{2} > \td{b}_\alpha
\end{equation}
where
\begin{equation}
\label{balpha}
    \td{b}_\alpha := \frac{|\ln{\alpha}| \sigma_0^2}{2\Delta} = \frac{|\ln{\alpha}| \sigma_0^2}{\eta - \mu_0},
\end{equation}
the stopping rule $\tau(\Lambda^{P_0,P_1^*}, b_\alpha)$ can be approximated by the stopping rule $\tau(\td{\Lambda}^{\mu_0,\eta}, \td{b}_\alpha)$, where 
\begin{equation}
\label{stat}
    \begin{split}
        \td{\Lambda}^{\mu_0,\eta} (t) &= \max_{1\leq k \leq t} \sum_{i=k}^t \left(X_i - \frac{\mu_0+\eta}{2}\right)\\
        &= \left(\td{\Lambda}^{\mu_0,\eta} (t-1)+ \left(X_t - \frac{\mu_0+\eta}{2}\right)\right)^+
    \end{split}
\end{equation}
with $\td{\Lambda}^{\mu_0,\eta}(0)=0$. We call $\tau(\td{\Lambda}^{\mu_0,\eta}, \td{b}_\alpha)$ the Mean-Change Test (MCT), and  $\td{\Lambda}^{\mu_0,\eta}$ the MCT statistic.

Noe that the worst-case delay satisfies
\begin{equation}
\label{est_opt_delay}
\begin{split}
    \inf_{\tau \in \mathcal{C}^{\{P_0\}}_{\alpha}} \sup_{P_1\in \mathcal{P}_1} \WADD{P_0,P_1}{\tau} &= \frac{|\ln{\alpha}|}{\KL{P_1^*}{P_0}}(1+o(1))\\
    &\approx \frac{|\ln{\alpha}| \sigma_0^2}{2\Delta^2}(1+o(1))
\end{split}
\end{equation}
where the approximation becomes more accurate as $\Delta \to 0$.

Therefore, if the gap is small, it is sufficient to know only the mean and variance to construct a good approximation to the asymptotically minimax robust test. Furthermore, only the mean of the pre-change distribution is needed to construct the MCT statistic. From the simulation results in Section IV, the performance of the MCT is very close to that of the asymptotically minimax robust test. Since the mean and variance of a distribution are much easier and more accurately estimated than the entire density, this test can be useful and accurate when only a moderate number of observations in the pre-change regime is available. 

\subsection{Performance Analysis of MCT for moderate $\Delta$}
We now study the asymptotic performance of the MCT  for the case where  $\Delta$ is moderate, as $\alpha \to 0$.

%(while still letting $\alpha \to 0$), using the same test statistic in (\ref{stat}). We investigate the case where both the mean and variance of the pre-change distribution are known.

% First, we recall the Bernstein inequality below \cite{all-nonpara-stat}:
% \begin{lemma}[Bernstein Inequality]
% Let $Y_1,\ldots,Y_n$ be independent, zero-mean random variables and $-M \leq Y_i \leq M$,
% \begin{equation}
%     P\left\{\sum_{i=1}^n Y_i > t\right\} \leq \exp{\left(-\frac{t^2}{2(v+Mt/3)}\right)}
% \end{equation}
% where $v=\sum_{i=1}^n \Var{}{}{Y_i}$.
% \end{lemma}

\begin{lemma}
\label{nonasym:fa}
Fix $P_0 \in \mc{P}_0$. For simplicity, denote $\tau := \tau(\td{\Lambda}^{\mu_0,\eta}, b)$, $Z_i := X_i - (\mu_0+\eta)/2$, and $S_t = \sum_{i=1}^t Z_i$. For any threshold $b > 1$,
\begin{equation}
\label{nng:fa}
\small
    \begin{split}
        P_0\left\{S_\tau \geq b\right\} &\leq 2 R_0 \sqrt{\frac{b^2}{\Delta^2}} K_1 \left(\frac{R_0^2 b \Delta }{\sigma_0^2} \right) \exp{\left(- \frac{R_0^2 \Delta}{\sigma_0^2} b \right)}\\
        &\sim \sqrt{\frac{2 \pi \sigma_0^2 b}{\Delta^3}} \exp{\left(- \frac{2 R_0^2 \Delta}{\sigma_0^2} b \right)}, \text{ as } b \to \infty
    \end{split}
\end{equation}
where
\begin{equation}
\label{rdef}
    R_0=\sigma_0^2 / \left(\sigma_0^2+\Delta \cdot \max\{\mu_0, 1-\mu_0\} / 3\right)
\end{equation}
and $K_\beta (z)$ is the modified Bessel function of the second kind with order $\beta$.
\end{lemma}

\begin{proof}
Let $M = \max\{\mu_0/3, (1-\mu_0)/3\}$. Note that $\E{}{P_0}{Z_i} = (\mu_0-\eta)/2 = -\Delta$. Thus, we have
\begin{align*}
    P_0\left\{S_\tau \geq b\right\} &= P_0\left\{\sum_{i=1}^\tau Z_i \geq b\right\} \\
    &= \sum_{t=1}^\infty P_0\left\{\sum_{i=1}^t Z_i \geq b, t = \tau\right\} \\
    &\leq \sum_{t=1}^\infty P_0\left\{\sum_{i=1}^t Z_i \geq b\right\} \\
    &= \sum_{t=1}^\infty P_0\left\{\sum_{i=1}^t (Z_i+\Delta) \geq b+t\Delta\right\} 
\end{align*}    
\begin{align}
    &\stackrel{(a)}{\leq} \sum_{t=1}^\infty \exp \left(-\frac{(b+t\Delta)^2}{2(t \sigma_0^2 + M (b+t\Delta))}\right) \nonumber\\
    &\stackrel{(b)}{\leq} \int_0^\infty \exp \left(-\frac{(b+x \Delta)^2}{2(x \sigma_0^2 + M (b+x \Delta))}\right) d x \nonumber\\
    &= a \int_0^\infty \exp \left(-\frac{(a \Delta y + C)^2}{2y}\right) d y \nonumber\\
    &= a e^{-a \Delta C} \int_0^\infty e^{-((a^2 \Delta^2 / 2)y + (C^2/2) y^{-1})} d y \nonumber\\
    &= \frac{2 C}{\Delta} e^{-a \Delta C} K_1(a \Delta C)
\end{align}
where $a := (\sigma_0^2+M\Delta)^{-1}$ and $C := \sigma_0^2 b / (\sigma_0^2 + M \Delta)$. In the series of inequalities above, (a) is by Bernstein's inequality \cite[p. 9]{all-nonpara-stat}, and (b) is from bounding the sum with an integral. Since $K_1(z) \sim \sqrt{\frac{\pi}{2 z}} e^{-z}$ as $|z| \to \infty$,
%when $\arg{z} = 0$
the asymptotic result follows. \qedhere
\end{proof}

\begin{theorem}
Let
\begin{equation}
\label{thr_set}
    % \td{b}'_\alpha := \frac{\td{b}_\alpha}{R_0^2} = \frac{|\ln\alpha|\sigma_0^2}{2 R_0^2 \Delta}
    \sqrt{\frac{2 \pi \sigma_0^2 \td{b}'_\alpha}{\Delta^3}} \exp{\left(- \frac{2 R_0^2 \Delta}{\sigma_0^2} \td{b}'_\alpha \right)} = \alpha
\end{equation}
then (\ref{main_fa_constraint}) is satisfied asymptotically as $\alpha \to 0$. Furthermore,
\begin{equation}
\label{thr_set2}
    \td{b}'_\alpha = \frac{\td{b}_\alpha}{R_0^2} (1+o(1))
\end{equation}
where $R_0$ is defined in (\ref{rdef}).
\end{theorem}

\begin{proof}
As $\alpha \to 0$, $\td{b}'_\alpha \to \infty$.
% Thus, equation (\ref{nng:fa}) becomes
% \begin{equation}
%     \begin{split}
%         \sqrt{\frac{2 \pi \sigma_0^2 \td{b}'_\alpha}{\Delta^3}} \exp{\left(- \frac{2 R_0^2 \Delta}{\sigma_0^2} \td{b}'_\alpha \right)} &= \sqrt{\frac{\pi \sigma_0^4 |\ln\alpha|}{\Delta^4 R_0^2}} e^{-|\ln\alpha|}\\
%         &= \alpha \sqrt{\frac{\pi \sigma_0^4 |\ln\alpha|}{\Delta^4 R_0^2}}\\
%         &= \alpha (1+o(1))
%     \end{split}
% \end{equation}
% because $|\ln\alpha|^{0.5} = o(\alpha)$. (?)
As a result, for any $P_1 \in \mc{P}_1$, $P_0\left\{S_\tau \geq b\right\} \leq \alpha (1+o(1))$. From \cite[Sec.~2.6]{siegmund_1985}, it can be shown that
\begin{equation}
\begin{split}
    \E{\infty}{P_0,P_1}{\tau} &= \frac{\E{}{P_0}{S_\tau}}{P_0\left\{S_\tau \geq b\right\}}\\
    &\geq \frac{1}{P_0\left\{S_\tau \geq b\right\}}\\
    &\geq \alpha^{-1}(1+o(1))
\end{split}
\end{equation}
since $\E{}{P_0}{S_\tau} \geq 1$. Thus, the false alarm constraint is satisfied asymptotically.

For the second result, it is sufficient to show that $(\td{b}'_\alpha - \td{b}_\alpha)/\td{b}_\alpha = R_0^{-2}-1+o(1)$. Let
\begin{equation}
    D := \frac{2 \Delta}{\sigma_0^2} \td{b}'_\alpha - |\ln \alpha|
\end{equation}
Note that the desired ratio $(\td{b}'_\alpha - \td{b}_\alpha)/\td{b}_\alpha = D/|\ln \alpha|$. Plugging into (\ref{nng:fa}), we have
\begin{equation}
    \sqrt{\frac{\sigma_0^4 \pi}{\Delta^4} (D+|\ln\alpha|)} e^{-R_0^2 (D+|\ln \alpha|)} = \alpha
\end{equation}
\begin{equation}
    -\frac{1}{2}\ln \left(\frac{\sigma_0^4 \pi}{\Delta^4} (D+|\ln\alpha|)\right) + R_0^2 (D+|\ln \alpha|) = |\ln\alpha|
\end{equation}
Now, we hypothesize that $D = D_1 |\ln \alpha| + o(|\ln \alpha|)$. The first term then vanishes because
\begin{equation}
\small
    \begin{split}
        \ln \left(\frac{\sigma_0^4 \pi}{\Delta^4} (D+|\ln\alpha|)\right) &= \ln \left(\frac{\sigma_0^4 \pi}{\Delta^4} ((D_1+1) |\ln \alpha| + o(|\ln \alpha|)\right)\\
        &= O(\ln(|\ln \alpha|))\\
        &= o(|\ln\alpha|)
    \end{split}
\end{equation}
From the second term, we validate our hypothesis that
\begin{equation}
    D = (R_0^{-2}-1) |\ln \alpha| + o(|\ln \alpha|)
\end{equation}
and the second result follows. 
\end{proof}

\begin{remark}
In practice, the threshold can be set to be $\sigma_0^2 / (2 R_0^2 \Delta)$ using equation (\ref{thr_set2}).
\end{remark}

\begin{theorem}
Fix $P_0 \in \mc{P}_0$. The worst-case delay satisfies
\begin{equation}
\label{est_stat}
    \sup_{P_1\in \mathcal{P}_1} \WADD{P_0,P_1}{\tau(\td{\Lambda}^{\mu_0,\eta}, \td{b}'_\alpha)} \leq \frac{|\ln{\alpha}| \sigma_0^2}{2\Delta^2 R_0^2}(1+o(1)) 
    % = \frac{|\ln{\alpha}| \sigma_0^2}{2\Delta^2}(1+o(1))
\end{equation}
as $\alpha \to 0$.
\end{theorem}

\begin{proof}

For any $P_1 \in \mc{P}_1$, as $\alpha \to 0$,
\begin{equation}
\begin{split}
    \WADD{P_0,P_1}{\tau(\td{\Lambda}^{\mu_0,\eta}, \td{b}'_\alpha)} &= \frac{\td{b}_\alpha'}{\Delta} (1+o(1))\\
    &= \frac{\td{b}_\alpha}{\Delta R_0^2} (1+o(1))\\
    &= \frac{|\ln{\alpha}| \sigma_0^2}{2\Delta^2 R_0^2}(1+o(1))
\end{split}
\end{equation}
where the first line is by renewal theory \cite{vvv_qcd_overview}. \qedhere
\end{proof}

\begin{remark}
As $\Delta \to 0$, $R_0 \to 1$. Thus, the result above becomes
\begin{equation}
    \sup_{P_1\in \mathcal{P}_1} \WADD{P_0,P_1}{\tau(\td{\Lambda}^{\mu_0,\eta}, \td{b}'_\alpha)} \sim \frac{|\ln{\alpha}| \sigma_0^2}{2\Delta^2}(1+o(1))
\end{equation}
where $o(1)$ goes zero as $\alpha$ and $\Delta$ go to zero, which coincides with the minimax robust worst-case delay in (\ref{est_opt_delay}).
\end{remark}

% We conclude this section by extending to the case as in (\ref{F0def}).

% \begin{corollary}
% All results hold by substituting $\sigma_0^2$ with $\mu_0 (1-\mu_0)$.
% \end{corollary}

% \begin{proof}
% It follows directly because the Bernstein Inequality only depends on up to the second moment.
% \end{proof}

\section{Numerical Results and Discussion}
 We study the performance of the proposed tests through simulations for the case where the pre- and post-change distributions are beta(4,16) ($\mu_0 = 0.2$) and beta(4.5,16) ($\mu_1 = 0.2195$), respectively. The mean-threshold $\eta$ is set to be $0.21$. In particular, we compare the performances for the following three statistics:
\begin{enumerate}
    \item The CUSUM statistic that knows both the pre- and post-change distributions, defined in (\ref{cusum_stat}).
    \item The statistic when only the pre-change distribution is known,  defined in (\ref{opt_stat}).
    \item The MCT statistic defined in (\ref{stat}).
\end{enumerate}
For all three statistics, based on their recursive structure, it is easy to show that the worst-case value of the change-point  for computing WADD  in \eqref{orig_qcd} is $\nu=0$. Therefore we can estimate the worst-case delay by simulating the post-change distribution from time 0.
% \vspace{-5mm}
\begin{figure}[htbp]
\centerline{\includegraphics[width=.5\textwidth,height=5.5cm]{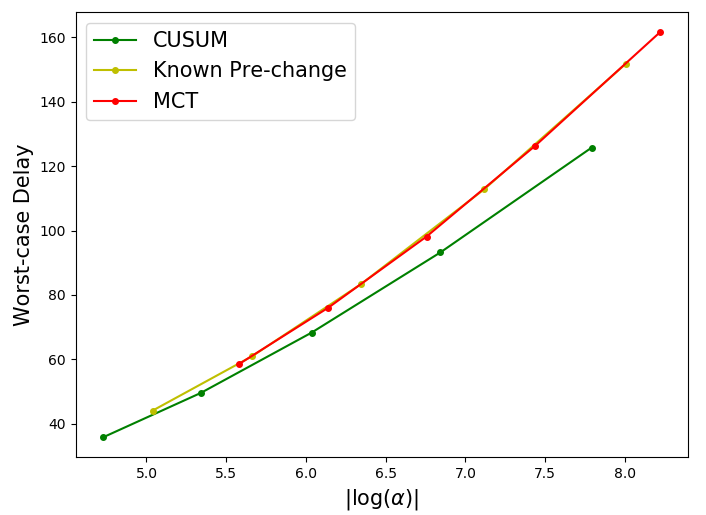}}
\vspace{-3mm}\caption{Performances of different statistics. The pre- and post-change distribution are Beta(4,16) ($\mu_0 = 0.2$) and Beta(4.5,16) ($\mu_1 = 0.2195$), respectively. The mean-threshold $\eta = 0.21$. .}
\label{fig1}
\end{figure}

%\vspace*{-3mm}
We see in Fig. \ref{fig1} that the  performance of MCT is very close to the test that uses knowledge of the entire pre-change distribution. Note that the MCT statistic only uses the pre-change mean; the variance is needed for setting the threshold to meet a given FAR constraint.
\begin{figure}[htbp]
\centerline{\includegraphics[width=.5\textwidth,height=5.4cm]{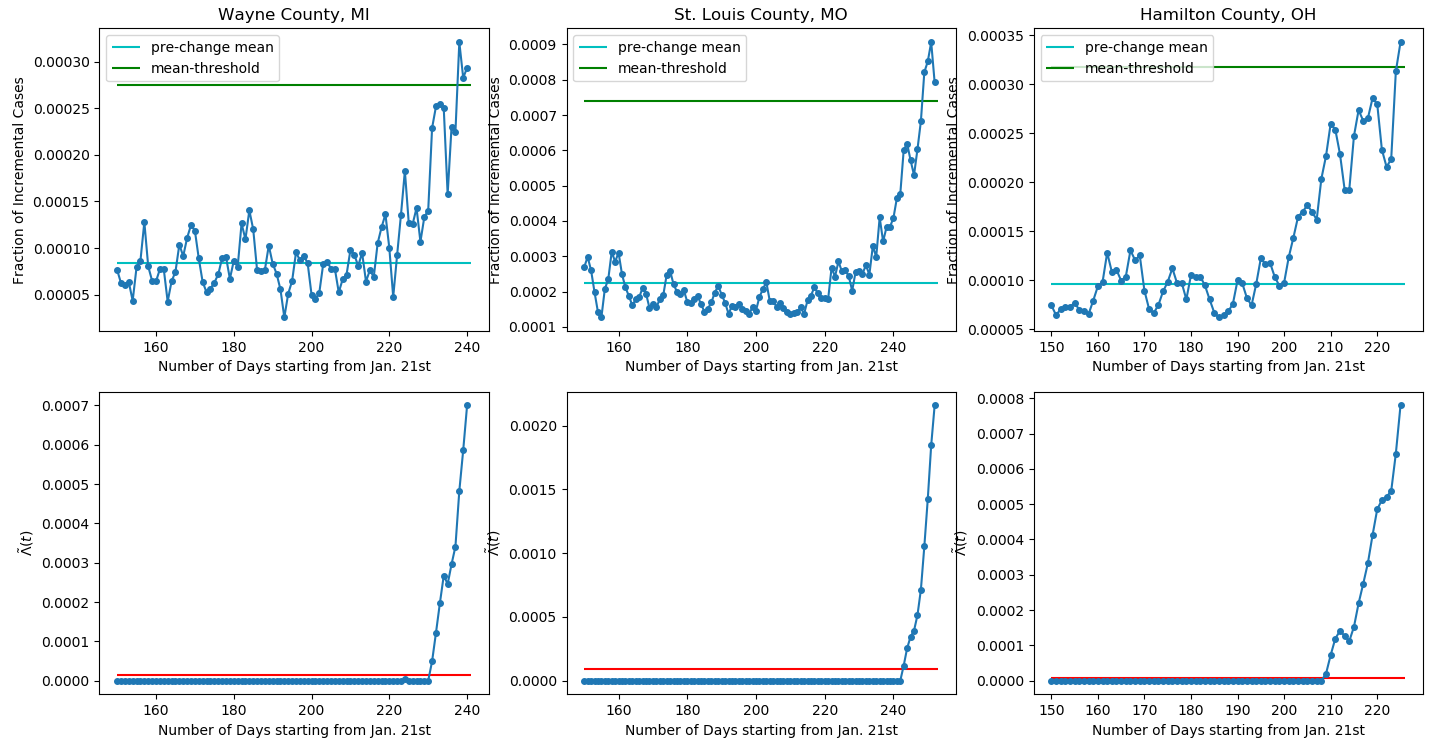}}
\vspace{-3mm}\caption{COVID-19 monitoring example. The upper subplot is the three-day moving average of the new cases  of COVID-19 as a fraction of the population in Wayne County, MI (left), St. Louis County, MO (middle), and Hamilton County, OH (right). The x-axis is the number days elapsed after January 21, 2020. The pre-change mean and variance are estimated using data from days 120 to 150. The FAR threshold $\alpha$ is set to $0.01$. For each county, the mean-threshold $\eta$ (in green) is set to be 3.3 times of the estimated pre-change mean (in cyan). The lower subplot shows the evolving of the statistic $\Tilde{\Lambda}$ in the corresponding county. The $\Lambda$-threshold $\td{b}_\alpha$ (in red) is calculated using equation (\ref{balpha}).}
\label{fig2}
\end{figure}
In Fig. \ref{fig2}, we apply the MCT to monitoring the spread of COVID-19 using new case data from various counties in the US \cite{nyt-covid-data}. The incremental cases from day to day can be assumed to be roughly independent.  The goal is to detect the onset of a new wave of the pandemic based on the incremental cases as a fraction of the county population exceeding some pre-specified level. 
The pre-change mean and variance are estimated using observations for periods in which the increments remain low and roughly constant. We set the mean-threshold $\eta$ to be a multiple of the pre-change mean, with understanding that such a threshold might be indicative of a new wave.  With this choice, we observe that 
the MCT statistic significantly and persistently crosses the test-threshold around late November in all counties, which is strong indication of a new wave of the pandemic. More importantly, unlike the raw observations which are highly varying, the MCT statistic shows a clear dichotomy between  the pre- and post-change settings, with the statistic staying near zero before the purported onset of the new wave, and taking off nearly vertically after the onset.
%

%For the first two counties, the MCT statistic crosses the $\Lambda$-threshold at Day 242 (left) and 242 (middle), respectively. 
%For Milwaukee County, the last time that the MCT statistic crosses the $\Lambda$-threshold happens on Day 240 (right). Although only few samples (if any) lie above the mean-threshold, the upward trend is captured by the algorithm and an alarm is made as soon as possible. Although the FAR is set to be 0.01, in simulation the actual FAR is usually around two orders of magnitude smaller. This shows that a second COVID-19 wave can be quickly detected using this test.

%\section{Conclusion}
%
%In this paper, we studied the problem of quickest detection of a change in the mean of an observation sequence in a non-parametric setting. We proposed low-complexity tests that approach the asymptotically optimal minimax detection delay as the false alarm rate goes to zero. In the future, it is interesting to study the change in statistics other than the mean. It is also rewarding to achieve a more accurate threshold for the change in the mean of bounded distributions.

\bibliographystyle{IEEEtran}
\bibliography{ref}

\end{document}